\documentclass[aps, prd, reprint]{revtex4-1}

\usepackage{amsfonts}
\usepackage{amssymb}
\usepackage{amsmath}
\usepackage{amsthm}
\usepackage{amscd}
\usepackage{graphicx}
\usepackage{hyperref}

\newcommand*{\dd}{\textrm{d}}

\newcommand*{\be}{\begin{equation}}
\newcommand*{\ee}{\end{equation}}
\newcommand*{\bea}{\begin{eqnarray}}
\newcommand*{\eea}{\end{eqnarray}}
\newcommand*{\bd}{\begin{displaymath}}
\newcommand*{\ed}{\end{displaymath}}

\newcommand*{\mt}{\textrm}
\newcommand*{\mc}{\mathcal}
\newcommand*{\mf}{\mathfrak}

\newcommand*{\mb}{\mathbb}
\newcommand*{\tr}{\textrm{tr}}

\newcommand*{\ad}{\textrm{ad}}

\newcommand*{\id}{\textrm{id}}

\newcommand*{\1}{\openone}

\theoremstyle{plain}
\newtheorem{propo}{Proposition}
\newtheorem{lemma}{Lemma}

\theoremstyle{definition}
\newtheorem{post}{Postulate}

\begin{document}

\title{Local Lorentz Covariance in Finite-dimensional Local Quantum Physics}
\date{\today}
\author{Matti Raasakka}
\email{mattiraa@gmail.com}
\affiliation{Tarhurintie 8 B 35, 01350 Vantaa, Finland}

\begin{abstract}
We show that local Lorentz covariance arises canonically as the group of transformations between local thermal states in the framework of Local Quantum Physics, given the following three postulates: (i) Local observable algebras are finite-dimensional. (ii) Minimal local observable algebras are isomorphic to $\mb{M}_2(\mb{C})$, the observable algebra of a single qubit. (iii) The vacuum restricted to any minimal local observable algebra is a non-maximally mixed thermal state. The derivation reveals a new and surprising relation between spacetime structure and local quantum states. In particular, we show how local restrictions of the vacuum can determine the connection between different local inertial reference frames.
\end{abstract}

\maketitle

\section{Introduction}
The Lorentz group is a fundamental structure in modern physics. In special relativity, it arises as the group of transformations between the global descriptions of physics by different inertial observers \cite{Wald84}. Together with the group of spacetime translations, it also forms the isometry group of spacetime in special relativity, the Poincar\'e group. Consequently, the structure of special relativistic quantum field theory (QFT), and the Standard Model of particle physics in particular, is dictated to a large extent by the requirement of invariance of measurement results under the Poincar\'e group \cite{Haag96}. In general relativity, spacetime is no longer Poincar\'e symmetric, but a local form of Lorentz covariance survives the generalization to curved spacetime manifolds \cite{Wald84}. The equivalence principle requires that spacetime should look approximately flat in any sufficiently small region. Therefore, the tangent spaces, which describe the local spacetime structure, inherit the action of the Lorentz group as the group of transformations between different local inertial reference frames. The covariance of physical quantities under local Lorentz transformations imposes fundamental constraints on the local physics in general relativity.

The Lorentz group has only infinite-dimensional non-trivial unitary representations \cite{Wigner39}. Since symmetries must be represented unitarily in quantum theory, it is necessary to introduce an infinite number of degrees of freedom in order to incorporate global Lorentz symmetry in quantum theory. However, in QFT even local bounded spacetime regions get assigned an infinite number of degrees of freedom due to the usage of fields. This leads to the infamous high energy divergencies in QFT. On the other hand, several well-known results, such as the finiteness of the black hole entropy \cite{Bekenstein73,Carlip14}, suggest that gravity should somehow regulate these divergencies. However, a naive physical energy cutoff in QFT breaks the local Lorentz covariance, which leads to severe problems with the physical plausibility of the regulated theory, given the fundamental role of the Lorentz group. Such considerations lead to the thorny puzzle of how to introduce a physical cutoff to the number of local degrees of freedom in QFT while preserving the local Lorentz covariance intact.

In this paper, we show that local Lorentz covariance arises canonically as the group of transformations between local thermal states in the framework of Local Quantum Physics, given three physically reasonable postulates about the local structure of quantum systems: (i) Local observable algebras are finite-dimensional. (ii) Minimal local observable algebras are isomorphic to $\mb{M}_2(\mb{C})$, the observable algebra of a qubit. (iii) The restriction of the vacuum state onto any minimal local observable algebra is a thermal state, which is not maximally mixed. The result follows from the mathematical fact, which we prove, that the proper orthochronous Lorentz group $SL(2,\mb{C}) / \mb{Z}_2$ is (isomorphic to) the group of transformations between such thermal states on $\mb{M}_2(\mb{C})$. Our result provides a natural way to incorporate local Lorentz covariance into Local Quantum Physics, which is compatible with and, indeed, follows from a physical cutoff to the number of local degrees of freedom. Furthermore, the form of the local Lorentz covariance reveals a new and surprising connection between quantum states and spacetime geometry, which we hope will lead to an improved understanding of quantum gravity. In particular, we demonstrate how local restrictions of the vacuum state can be used to determine the connection between different local inertial reference frames.

\section{Local quantum physics}
Local Quantum Physics (LQP) is a well-established algebraic approach to rigorously define QFT models \cite{Haag96}. The generally covariant formulation \cite{BrunettiFredenhagenVerch03} applies to arbitrary globally hyperbolic spacetimes. The formulation of a quantum theoretical model according to the LQP prescription begins by associating to each local bounded region $\mc{O}\subset\mc{M}$ of spacetime $\mc{M}$ the algebra of quantum observables $\mf{A}(\mc{O})$ localized in that region, typically assumed to be a unital $C^*$-algebra. These local observable algebras are required to satisfy the \emph{isotony} property: if $\mc{O}_1 \subset \mc{O}_2$ is a proper inclusion, then $\mf{A}(\mc{O}_1)$ is a proper unital $C^*$-subalgebra of $\mf{A}(\mc{O}_2)$. We may define the quasi-local observable algebra containing all the local observables as
\begin{displaymath}
	\mf{A}(\mc{M}) := \overline{ \cup_{\mc{O}\subset\mc{M}} \mf{A}(\mc{O}) } \,,
\end{displaymath}
where the overline denotes the $C^*$-norm completion. A global quantum state on the quasi-local observable algebra is then given by a consistent assignment of local states $\omega_\mc{O} : \mf{A}(\mc{O}) \rightarrow \mb{C}$ to each of the local algebras $\mf{A}(\mc{O})$, $\mc{O}\subset\mc{M}$, such that if $\mc{O}_1 \subset \mc{O}_2$, then the restriction of $\omega_{\mc{O}_2}$ onto $\mf{A}(\mc{O}_1)$ agrees with $\omega_{\mc{O}_1}$, i.e., $\left. \omega_{\mc{O}_2} \right|_{\mf{A}(\mc{O}_1)} = \omega_{\mc{O}_1}$.

Another important property the local observable algebras must satisfy is \emph{microlocality}: If $\mc{O}_1$ is spacelike separated from $\mc{O}_2$, then $\mf{A}(\mc{O}_1)$ commutes with $\mf{A}(\mc{O}_2)$. Microlocality guarantees the joint measurability of spacelike separated observables. On the other hand, if $\mf{A}(\mc{O}_1)$ and $\mf{A}(\mc{O}_2)$ do not mutually commute, they are statistically dependent in the sense that all states have correlations over them (i.e., there are no product states) \cite{Buchholz04b}.

\section{Finite-dimensional local algebras}
For physical QFT models the local observable algebras are infinite-dimensional hyperfinite factors of type III$_1$ \cite{Buchholz87}. However, as initially proposed by Haag \cite{Haag96}, we want to consider the case that the local observable algebras are finite-dimensional. 
\begin{post}\label{post:finite}
	Local observable algebras are finite-dimensional $C^*$-algebras.
\end{post}
The motivation for Postulate \ref{post:finite} comes from the following two points: (i) As mentioned in the introduction, QFT suffers from high energy divergencies, which may be cured by limiting the number of local degrees of freedom. (See, e.g., \cite{Pye15}.) In LQP, such a physical regularization is particularly simple to implement by requiring the local observable algebras to be finite-dimensional. We expect QFT to appear as an effective field theory obtained from some finite-dimensional model in the infinite-dimensional limit, where the dimension of the local algebras is taken to approach infinity, and the divergencies may reappear. (ii) The entropy of any subsystem in QFT is generically divergent. However, results on black hole thermodynamics suggest that black holes \cite{Carlip14}, and possibly any subsystem \cite{Bousso99}, should be endowed with a finite amount of entropy proportional to the area of its spatial boundary. If a local observable algebra has dimension $n \in \mb{N}$, then the von Neumann entropy of the local states on that algebra is bounded from above by $\sim \ln(n)$. Accordingly, the finite-dimensionality of the local observable algebras regulates the local entropy. For low energy states, the leading divergence of the entropy in the infinite-dimensional limit should then scale generically as the boundary area, as in QFT \cite{Srednicki93,Casini14}, since the low energy states are insensitive to the microscopic structure of the theory.

Let us remark that, due to microlocality,
\begin{displaymath}
	\mf{A}(\mc{O}_1 \cup \mc{O}_2) \cong \mf{A}(\mc{O}_1) \otimes \mf{A}(\mc{O}_2) \,,
\end{displaymath}
when the two regions $\mc{O}_1, \mc{O}_2 \subset \mc{M}$ are spacelike separated. The quantity $\ln\dim(\mf{A}(\mc{O}))$ is extensive under tensor product and measures the maximum amount of information that can be contained in $\mc{O}$. We expect this information carrying capacity of $\mc{O}$ to be proportional to the 4-volume of the domain of dependence of $\mc{O} \in \mc{M}$, so that the density of degrees of freedom is constant in spacetime. Accordingly, given a causally complete spacetime region $\mc{O}$, we may write
\begin{displaymath}
	V(\mc{O}) = v \ln\dim(\mf{A}(\mc{O}))
\end{displaymath}
for the 4-volume of $\mc{O}$, where $v\in\mb{R}_+$ is a constant with the dimensions of a 4-volume. We suspect that the constant $v$ determining the density of degrees of freedom in spacetime should be proportional to the Planck volume $v_{\mt{Pl}} = (\hbar G / c^3)^2 \approx 10^{-140}m^4$. If we keep $v$ fixed, the infinite-dimensional limit corresponds to the infinite-volume thermodynamical limit. On the other hand, we may take $v \rightarrow 0$ in the infinite-dimensional limit in order to describe physical systems with finite 4-volume and infinite density of degrees of freedom, as in QFT. Notice that this `field theory limit' still describes the physics in any measurable-sized spacetime volume extremely accurately due to the smallness of $v$. If we set $v \propto v_{\mt{Pl}}$, the limit $v \rightarrow 0$ can correspond either to $\hbar \rightarrow 0$, $G \rightarrow 0$, or $c \rightarrow \infty$, which suggests a fundamental inconsistency of local QFT and quantum gravity: In order to arrive at a local relativistic ($c < \infty$) quantum ($\hbar > 0$) model in the field theory limit, we must neglect gravitational interactions ($G \rightarrow 0$). One possible way out of this dilemma is to consider the finite-dimensional local description to give the fundamental definition of the theory, which gives us a further incentive to study finite-dimensional LQP.

\section{Minimal local observable algebras}
According to the isotony property of the local observable algebras, smaller spacetime regions should correspond to algebras of smaller dimensionality. Assuming that the local observable algebras that correspond to causally complete spacetime regions are factors (i.e., have trivial centers), as in QFT \cite{Buchholz87}, we find that the minimal local observable algebras must be isomorphic to $\mb{M}_p(\mb{C})$, the algebra of $p$-by-$p$ complex-valued matrices,  for some $p$ prime: None of these can be included into one another as unital subalgebras, since such inclusions are of the form
\begin{displaymath}
	\mb{M}_p(\mb{C}) \mapsto \mb{M}_p(\mb{C}) \otimes \1_k \subset \mb{M}_{pk}(\mb{C}) \,,
\end{displaymath}	
where $\1_k$ is the identity matrix in $\mb{M}_k(\mb{C})$ for some $k\in\mb{N}$, and obviously $pk\in\mb{N}$ cannot be prime. The smallest non-trivial such algebra is $\mb{M}_2(\mb{C})$, and therefore we postulate:
\begin{post}\label{post:qubit}
	Minimal local observable algebras are isomorphic to $\mb{M}_2(\mb{C})$.
\end{post}
$\mb{M}_2(\mb{C})$ is exactly the observable algebra of a single qubit, which is a system capable of storing a minimal amount of (quantum) information.

\section{Local thermal states}
In Minkowski spacetime, the global vacuum state is characterized by its minimal energy. However, restricted to local subregions, the vacuum gives rise to thermal states \cite{Haag96}. Due to the equivalence principle, according to which any small enough spacetime region is well-approximated by a piece of the Minkowski spacetime, we expect the vacuum to be approximately thermal also in any sufficiently small region of a generic curved spacetime. By extension, we postulate:
\begin{post}\label{post:localeq}
	The vacuum restricted to any minimal local observable algebra is a thermal state, which is not maximally mixed.
\end{post}

A thermal state of a finite-dimensional quantum system is represented by a density matrix of the Gibbs form
\begin{equation}\label{eq:gibbs}
	\rho = e^{-\beta H} / \mc{Z}_{\beta,H} \,,
\end{equation}
where $\mc{Z}_{\beta,H} := \tr(e^{-\beta H})$ is the partition function, $H$ is the Hamiltonian operator, and $\beta \in \mb{R}_+$ is the inverse temperature. The maximally mixed state corresponds to $\rho = \frac{1}{2}\1_2$. In Appendix \ref{app:propo1} we prove:
\begin{propo}\label{propo:gibbs}
Let $\rho\in\mb{M}_2(\mb{C})$ be a non-degenerate density matrix, and $\rho \neq \frac{1}{2}\emph{\1}_2$. Then, at any fixed inverse temperature $\beta\in\mb{R}_+$, there exists a unique Hamiltonian $H$ with spectrum $\{0,\epsilon > 0\}$ satisfying (\ref{eq:gibbs}).
\end{propo}

In fact, any non-degenerate density matrix in $\mb{M}_2(\mb{C})$ can be written in the form (\ref{eq:gibbs}). However, the physical property that $\rho$ is thermal at the inverse temperature $\beta$ implies that $H$ measures the physical energy content of the system. For the restricted vacuum state on the minimal local observable algebras $\mf{A}(\mc{O}) \cong \mb{M}_2(\mb{C})$, we therefore interpret the unique operator $H$ provided by Proposition \ref{propo:gibbs} to measure the local energy content in the associated minimal 4-volume $V(\mc{O}) = v\ln(4)$. The restriction $\rho\neq \frac{1}{2}\1_2$ is motivated by the requirement that $\rho$ is a finite temperature thermal state with respect to a non-zero local Hamiltonian, so that local excitations carry non-zero energy. In the following, we restrict to consider non-maximally mixed (NMM) thermal states.

\section{Local Lorentz covariance}
Thermal states are generally not Lorentz invariant, since the Hamiltonian $H$ must measure the energy content of a system with respect to some particular inertial reference frame \cite{Ojima86}. Therefore, the Lorentz group must act non-trivially on the local thermal states on the minimal local observable algebras. Indeed, the proper orthochronous Lorentz group is (isomorphic to) $SL(2,\mb{C}) / \mb{Z}_2$, which acts canonically on the hermitian elements $H\in \mb{M}_2(\mb{C})$ as \cite{Wald84}
\begin{equation}\label{eq:SL2Caction}
	H \mapsto \Lambda H \Lambda^* \,,\ \Lambda \in SL(2,\mb{C}) \,.
\end{equation}
We may take $SL(2,\mb{C}) / \mb{Z}_2$ to act similarly on the local Hamiltonian on any $\mf{A}(\mc{O}) \cong \mb{M}_2(\mb{C})$. To justify the interpretation of the action (\ref{eq:SL2Caction}) of $SL(2,\mb{C})$ on a local Hamiltonian as a Lorentz transformation, we note that the eigenvalues of $H$ transform as the timelike component of a Lorentz 4-vector under this implementation of the Lorentz group, as is appropriate for energy: Parametrize the local Hamiltonian as
\begin{equation}\label{eq:Hparam}
	H = h^0 \1_2 + \sum_{i=1}^3 h^i \sigma_i \,,
\end{equation}
where $h^\mu \in\mb{R}$, $\mu=0,\ldots,3$, and $\sigma_i$, $i=1,2,3$, are the Pauli matrices. The eigenvalues of $H$ in terms of $h^\mu$ are $h^0 \pm \sqrt{\sum_i (h^i)^2}$. Since the eigenvalues of $H$ are $0$ and $\epsilon > 0$, we find $h^0 = \sqrt{\sum_i (h^i)^2}$ and $\epsilon = 2h^0 \equiv \tr(H)$. The parameters $h^\mu$ transform as a 4-vector under the action (\ref{eq:SL2Caction}) of $SL(2,\mb{C})$ \cite{Wald84}, $h^0$ being the timelike component. Thus, the eigenvalues $0$ and $\epsilon$ of $H$ transform as the timelike components of the 4-vectors $0$ and $2h^\mu$, respectively.

The map (\ref{eq:SL2Caction}) acts on thermal density matrices as
\begin{equation}\label{eq:lorentz}
	e^{-\beta H} / \mc{Z}_{\beta,H} \stackrel{\Lambda}{\mapsto} e^{-\beta (\Lambda H \Lambda^*)} / \mc{Z}_{\beta,\Lambda H \Lambda^*} \,.
\end{equation}
By the above transformation properties of $H$, the action (\ref{eq:lorentz}) maps NMM thermal states to each other. Furthermore, we can also show:
\begin{propo}\label{propo:lorentz}
Any two NMM thermal states on $\mb{M}_2(\mb{C})$ can be mapped to each other by the action (\ref{eq:lorentz}) of $SL(2,\mb{C})$.
\end{propo}
\begin{proof}
Let the two states on $\mb{M}_2(\mb{C})$ be given by the density matrices $\rho_a = e^{-\beta_a H_a} / \mc{Z}_{\beta_a,H_a}$, $a=1,2$. By Proposition \ref{propo:gibbs} we may choose $\beta_1 = \beta_2 \equiv \beta$ and $H_a$ with spectra $\{0,\epsilon_a > 0\}$. Then, the 4-vectors $h_a^\mu$ parametrizing $H_a$ according to (\ref{eq:Hparam}) are forward lightlike, i.e., $h_a^0 > 0$ and $(h_a^0)^2 - \sum_i (h_a^i)^2 = 0$. Any two such 4-vectors can be mapped to each other by a proper orthochronous Lorentz transformation. Since (\ref{eq:SL2Caction}) can give rise to any such transformation, the claim is true.
\end{proof}
In summary, (i) any Lorentz transformation (\ref{eq:lorentz}) maps a NMM thermal state on $\mb{M}_2(\mb{C})$ to another NMM thermal state, i.e., the map (\ref{eq:lorentz}) cannot take us outside the set of NMM thermal states, and (ii) any two NMM thermal states on $\mb{M}_2(\mb{C})$ are mapped to each other by a Lorentz transformation (\ref{eq:lorentz}). Also, $\1_2\in SL(2,\mb{C})$ is the only element leaving all the NMM thermal states fixed. In this way, the Lorentz group appears canonically as the group of transformations between non-maximally mixed thermal states on any minimal local observable algebra.

\section{Spacetime geometry from vacuum}
Postulate \ref{post:localeq} and Proposition \ref{propo:lorentz} together imply that the vacuum is homogeneous throughout spacetime, up to local Lorentz transformations. Therefore, we may interpret the change in the local Hamiltonian from one minimal local observable algebra to another as indicating that the two algebras occupy different local inertial reference frames. Conversely, we may infer relations between the local inertial reference frames of the minimal local observable algebras by comparing their local Hamiltonians.

For a simple demonstration of how this recovery of the connection between local inertial reference frames works, let us consider the inclusion of a minimal local observable algebra $\mf{A}(\mc{O}) \cong \mb{M}_2(\mb{C})$ into a larger local observable algebra $\mf{B} \cong \mb{M}_{2n}(\mb{C})$ for some $n\in\mb{N}$. Concretely, we may write this inclusion as $a \mapsto a\otimes \1_n \in \mb{M}_{2n}(\mb{C})$ for all $a\in\mb{M}_2(\mb{C})$, where $\1_n$ is the identity matrix in $\mb{M}_n(\mb{C})$. Denote the image of $\mf{A}(\mc{O})$ under this inclusion as $\mf{A}_1 = \mb{M}_2(\mb{C}) \otimes \1_n \subset \mb{M}_{2n}(\mb{C})$. Let us also consider another inclusion of $\mb{M}_2(\mb{C})$, which is related to the first one by an infinitesimal unitary transformation:
\begin{equation}\label{eq:infuni}
	\mf{A}_2 = e^{i\epsilon X}(\mb{M}_2(\mb{C}) \otimes \1_n)e^{-i\epsilon X} \subset \mb{M}_{2n}(\mb{C})
\end{equation}
for some hermitian $X \in \mb{M}_{2n}(\mb{C})$ and $\epsilon \ll 1$. At least for some $X$, we may think of $\mf{A}_2$ as obtained from $\mf{A}_1$ by an infinitesimal spatiotemporal displacement, because (i) any minimal spatiotemporal subvolume of the larger system $\mf{B}$ corresponds to a subalgebra of $\mf{B}$ isomorphic to $\mb{M}_2(\mb{C})$, and (ii) any two isomorphic subalgebras of $\mf{B} \cong \mb{M}_{2n}(\mb{C})$ are related by a unitary transformation.

Now, let the vacuum restrict on $\mf{B}$ to a state represented by the density matrix $\rho\in\mb{M}_{2n}(\mb{C})$. Then, the restrictions of $\rho$ onto $\mf{A}_1$ and $\mf{A}_2$ are given by the reduced density matrices
\begin{eqnarray}
	\rho_1 &=& (\id_2 \otimes \tr_n)(\rho) \in \mb{M}_2(\mb{C}) \,,\nonumber\\
	\rho_2 &=& (\id_2 \otimes \tr_n)(e^{-i\epsilon X} \rho e^{i\epsilon X}) \in \mb{M}_2(\mb{C}) \,,\nonumber
\end{eqnarray}
respectively, where $(\id_2 \otimes \tr_n)$ denotes the partial trace over the second tensor product factor in the decomposition $\mb{M}_{2n}(\mb{C}) = \mb{M}_2(\mb{C}) \otimes \mb{M}_n(\mb{C})$. Accordingly, we find
\begin{displaymath}
	 \left.\frac{\dd}{\dd \epsilon}\right|_{\epsilon = 0} \rho_2 = - i (\id_2 \otimes \tr_n)([X, \rho]) =: \delta_X\rho_1
\end{displaymath}
for the perturbation of $\rho_1$ under the infinitesimal unitary transformation (\ref{eq:infuni}) of $\mf{A}_1$.

On the other hand, we may consider an infinitesimal Lorentz transformation of $\rho_1$. By Postulate \ref{post:localeq} and Proposition \ref{propo:gibbs}, $\rho_1 = e^{-\beta H_1} / \mc{Z}_{\beta,H_1}$ for some $\beta \in \mb{R}_+$ and $H_1$ such that $\tr(H_1) > 0$, $\det(H_1) = 0$. An infinitesimal Lorentz transformation of the form (\ref{eq:lorentz}) is induced by $\Lambda_\epsilon = e^{\epsilon Z} \in SL(2,\mb{C})$, where $Z \in \mb{M}_2(\mb{C})$ is traceless. Then, from the equation
\begin{equation}\label{eq:unitdisp}
	\left.\frac{\dd}{\dd\epsilon}\right|_{\epsilon=0} \frac{e^{-\beta (\Lambda_\epsilon H_1 \Lambda_\epsilon^*)}}{\mc{Z}_{\beta,\Lambda_\epsilon H_1 \Lambda_\epsilon^*}} = \left.\frac{\dd}{\dd \epsilon}\right|_{\epsilon = 0} \rho_2 \equiv \delta_X\rho_1
\end{equation}
we may solve for $Z$ in order to find out, which infinitesimal Lorentz transformation the unitary displacement (\ref{eq:infuni}) of $\mf{A}_1$ in the direction of $X$ induces.

The derivation of the general solution for $Z$ is deferred to Appendix \ref{app:eq7}. There, we find
\begin{equation}\label{eq:Z}
	Z = \frac{1}{\tr(H_1)^2} \left(\tau \Big(H_1 - \frac{1}{2}\tr(H_1)\1_2 \Big) + [\delta_X H_1, H_1] \right) \,,
\end{equation}
where we use the following notations:
\begin{eqnarray}
	\delta_X H_1 &:=& -\beta^{-1} \rho_1^{-\frac{1}{2}} \mc{L}_{\beta H_1}(\delta_X \rho_1) \rho_1^{-\frac{1}{2}} \,, \nonumber\\
	\tau &:=& \frac{2\tr(H_1 \delta_X H_1) - \tr(H_1)\tr(\delta_X H_1)}{\tr(H_1)} \,, \nonumber\\
	\mc{L}_{\beta H_1} &:=& \frac{\ad_{\beta H_1}/2}{\sinh(\ad_{\beta H_1}/2)} \,. \nonumber
\end{eqnarray}
Here, $\mc{L}_{\beta H_1}$ is a linear operator on $\mb{M}_2(\mb{C})$ defined in terms of $\ad_{\beta H_1}(Y) := \beta [H_1,Y]$ by the Taylor expansion
\begin{displaymath}
	\mc{L}_{\beta H_1}(Y) = \sum_{n=0}^\infty \frac{\beta^n}{n!} \left[ \frac{\dd^n}{\dd x^n}\left( \frac{x/2}{\sinh(x/2)} \right) \right]_{x=0} (\ad_{H_1})^n(Y) \,,
\end{displaymath}
where $(\ad_{H_1})^n$ denotes the $n$-fold composition of $\ad_{H_1}$.

The remarkable consequence of the general solution (\ref{eq:Z}) is that the vacuum state can encode properties of the local spacetime geometry by providing the connection between the local inertial reference frames associated to infinitesimally different minimal spatiotemporal volumes.

\section{Summary \& Discussion}
We have shown that, given our Postulates \ref{post:finite}, \ref{post:qubit} and \ref{post:localeq}, local Lorentz covariance appears in finite-dimensional Local Quantum Physics as transformations between thermal states on the minimal local observable algebras. Moreover, we demonstrated how the vacuum can encode local spacetime geometry by providing the connection between the local inertial reference frames of infinitesimally different minimal spacetime regions. Our results provide a new and surprising relation between spacetime structure and local quantum states, which we hope will lead to an improved understanding of quantum gravity.

The general idea of deriving spacetime symmetries from elementary quantum theory is not new, of course; see, e.g., \cite{Goernitz92}. More recent work \cite{Hoehn14} provides a derivation of the local Lorentz covariance from quantum information theory. The approach of \cite{Hoehn14} is complementary to ours in the sense that it focuses on the communication relations between local observers instead of the local properties of the vacuum. In our view, the main strength of our approach as compared to \cite{Hoehn14} is its direct relation to QFT via finite-dimensional LQP, which should allow for a more immediate application to quantum gravity.

It is also worth noting that our derivation, as that of \cite{Hoehn14}, is specific to the $3+1$-dimensional Lorentz group. This limitation does not worry us too much at the moment, as all experimental data to date is consistent with four spacetime dimensions. In fact, it can also be considered to be a positive feature, since it means that the framework is immediately ruled out, if a larger spacetime symmetry group is found. Falsifiability of a theoretical idea should always be a positive thing in science.

In \cite{Raasakka16} we introduced a spacetime-free framework for quantum theory, which is based on the construction of the kinematical quantum observable algebra ``from the ground up'' as the free product of component algebras corresponding to individual observables. The physical observable algebra was suggested to be obtained via the GNS representation induced by some reference state, such as the vacuum. An important application of this framework will be the construction of finite-dimensional LQP models by applying the free product construction of \cite{Raasakka16} to glue collections of minimal local observable algebras together in different ways, and studying the quantum systems that arise in this way. The results in this paper will allow to directly associate spacetime geometries to the reference states of such systems.

Ultimately, our goal is to develop a theoretical framework capable of accommodating both QFT and gravity. A more thorough investigation of finite-dimensional LQP from this perspective is ongoing. As already mentioned, we expect QFT to emerge out of finite-dimensional LQP in the field theory limit. Therefore, the explicit definition and study of the infinite-dimensional limits of finite-dimensional LQP models will be vital to our future work. Gravity, on the other hand, should perhaps be expected to manifest itself only outside this limit, as we have argued above. Following e.g.\ \cite{Jacobson95,Jacobson15}, we suspect that gravity is a quantum statistical effect. Finite-dimensional LQP seems to provide a suitable and well-behaved framework to further study this idea. We look forward to explore its relation with the connection of local quantum states to spacetime geometry uncovered in this work.

\begin{acknowledgments}
We would like to thank Paolo Bertozzini, Philipp H\"ohn and the anonymous referees for their helpful comments on the manuscript.
\end{acknowledgments}

\appendix
\section{Proof of Proposition 1}\label{app:propo1}
Let us first prove the following lemma.
\begin{lemma}\label{lemma:K}
A non-degenerate density matrix $\rho \in \mb{M}_2(\mb{C})$ can always be written in the form
\be\label{eq:gibbsK}
	\rho = \frac{e^{-K}}{\emph{\tr}(e^{-K})}
\ee
such that the spectrum of $K$ is $\{0,\epsilon \geq 0\}$. Moreover, $K$ satisfying these requirements is uniquely given by
\be\label{eq:K}
	-\ln \rho - \lambda \emph{\1}_2 \,,
\ee
where $\lambda \in \mb{R}_+$ is the smallest eigenvalue of $-\ln \rho > 0$.
\end{lemma}
\begin{proof}
Notice that $K = -\ln\rho - \lambda \1_2$ satisfies (\ref{eq:gibbsK}) for any $\lambda\in\mb{R}$. Moreover, we can easily verify that all matrices $K$ satisfying (\ref{eq:gibbsK}) for a fixed $\rho$ are of this form by taking logarithms on both sides of (\ref{eq:gibbsK}). Then, $\det(K) = 0$ gives
\bd
	\det(-\ln\rho - \lambda\1_2) = 0 \,,
\ed
the characteristic equation of $-\ln\rho$, which has two solutions: the eigenvalues $\lambda_i > 0$ of $-\ln\rho$. (Notice that $-\ln\rho > 0$, as $0 < \rho < 1$.) Let $\lambda_1 \leq \lambda_2$. Then, $-\ln\rho - \lambda_1\1_2 \geq 0$ and $-\ln\rho - \lambda_2\1_2 \leq 0$. Accodingly, (\ref{eq:K}) is the unique solution for $K$ satisfying (\ref{eq:gibbsK}) with spectrum $\{0, \lambda_2 - \lambda_1 =: \epsilon \geq 0\}$.
\end{proof}
The following is a simple consequence of Lemma \ref{lemma:K}:
\begin{lemma}\label{lemma:gibbs}
Let $\rho\in\mb{M}_2(\mb{C})$ be a non-degenerate density matrix. Then, the inverse temperature $\beta\in\mb{R}_+$ and the Hamiltonian $H$ with spectrum $\{0,\epsilon \geq 0\}$, satisfying
\be\label{eq:gibbs2}
	\rho = \frac{e^{-\beta H}}{\emph{\tr}(e^{-\beta H})}
\ee
are unique up to the simultaneous scaling
\bd
	\beta \mapsto \lambda\beta \,,\quad H \mapsto \lambda^{-1} H \,,\quad \lambda \in \mb{R}_+ \,.
\ed
\end{lemma}
\begin{proof}
By comparing (\ref{eq:gibbs2}) to (\ref{eq:gibbsK}), Lemma \ref{lemma:K} shows that the product $\beta H$ in (\ref{eq:gibbs2}) is unique. Since the product $\beta H$ is invariant under the scaling $\beta \mapsto \lambda \beta$, $H \mapsto \lambda^{-1}H$ for any $\lambda\in\mb{R}_+$, we are still left with this freedom in the choice of $\beta$ and $H$.
\end{proof}
Thus, fixing $\beta\in\mb{R}_+$ leads to a unique Hamiltonian by Lemma \ref{lemma:gibbs}. With the restriction $\rho \neq \frac{1}{2}\1_2$ corresponding to $H\neq 0 \Leftrightarrow \epsilon > 0$ this proves Proposition 1.

\section{Derivation of Equation 7}\label{app:eq7}
Let us first prove the following helpful lemma.
\begin{lemma}\label{lem:Hpert}
Let $\rho \in \mb{M}_2(\mb{C})$ be a non-degenerate density matrix. Then, a linear perturbation
\bd
	\rho \mapsto \rho + \epsilon \delta\rho + \mc{O}(\epsilon^2)
\ed
for $\epsilon \ll 1$  and traceless $\delta\rho\in\mb{M}_2(\mb{C})$ corresponds to a linear perturbation of the thermal Hamiltonian given by
\bd
	H \mapsto H + \epsilon \delta H + \mc{O}(\epsilon^2) \,,
\ed
where
\bd
	\delta H = -\beta^{-1} \rho^{-\frac{1}{2}} \frac{\emph{\ad}_{\beta H}/2}{\sinh(\emph{\ad}_{\beta H}/2)}(\delta\rho) \rho^{-\frac{1}{2}} \,.
\ed
\end{lemma}
\begin{proof}
Remember that $\rho = e^{-\beta H}/\tr(e^{-\beta H})$, where $H$ is the thermal Hamiltonian. Up to the first order in $\epsilon$,
\bd
	\rho + \epsilon \delta\rho = \frac{e^{-\beta H} e^{\epsilon \rho^{-1} (\delta\rho)}}{\tr(e^{-\beta H}e^{\epsilon \rho^{-1} (\delta\rho)})} \,,
\ed
as can be verified by expanding the exponentials in $\epsilon$, and noting that $\tr(\delta\rho) = 0$. Let us define $\delta H\in\mb{M}_2(\mb{C})$ via the requirement $e^{-\beta H} e^{\epsilon \rho^{-1} (\delta\rho)} \equiv e^{-\beta(H + \epsilon \delta H)}$. By \cite[Theorem 3.5]{Hall03}, any smooth function $\zeta: \mb{R} \rightarrow \mb{M}_2(\mb{C})$ satisfies for sufficiently small $\epsilon$
\bea
	e^{-\zeta(\epsilon)} \frac{\dd}{\dd \epsilon} e^{\zeta(\epsilon)} &=& \sum_{n=0}^\infty \frac{(-1)^n (\ad_{\zeta(\epsilon)})^n}{(n+1)!} \left( \frac{\dd \zeta}{\dd \epsilon}(\epsilon) \right) \nonumber\\
	&\equiv & \left[ \frac{1 - e^{-\ad_{\zeta(\epsilon)}}}{\ad_{\zeta(\epsilon)}} \right] \left( \frac{\dd \zeta}{\dd \epsilon}(\epsilon) \right) \,, \label{eq:BCH}
\eea
where $\ad_{X}(Y) := [X,Y]$ for any $X,Y \in \mb{M}_2(\mb{C})$, we use the notation
\bd
	(\ad_X)^n = \underbrace{\ad_X \circ \cdots \circ \ad_X}_{n \mt{ pcs}} \,,
\ed
and the formal expression on the last line is defined in terms of the expansion on the previous line. Choosing $\zeta(\epsilon) = -\beta(H + \epsilon \delta H)$ gives $\zeta(0) = -\beta H$ and $\frac{\dd \zeta}{\dd \epsilon}(\epsilon) = -\beta\delta H$. On the other hand, we find
\bd
	\left(e^{-\beta H} e^{\epsilon \rho^{-1} (\delta\rho)}\right)^{-1} \frac{\dd}{\dd \epsilon} e^{-\beta H} e^{\epsilon \rho^{-1} (\delta\rho)} =  \rho^{-1}(\delta\rho)
\ed
at $\epsilon = 0$. Thus, we get from (\ref{eq:BCH}) the relation
\bd
	\rho^{-1}(\delta\rho) = \left[ \frac{e^{\ad_{\beta H}} - 1}{\ad_{\beta H}} \right] \left( -\beta\delta H \right) \,,
\ed
and finally acting on both sides by the inverse operator $\ad_{\beta H}/(e^{\ad_{\beta H}} - 1)$, and using $\rho^{-\frac{1}{2}}(\delta\rho)\rho^{\frac{1}{2}} = e^{-\ad_{\beta H}/2}(\delta\rho)$, we find the stated result.
\end{proof}
Any state is invariant under the addition of terms proportional to the identity to the thermal Hamiltonian, and therefore also
\bd
	H \mapsto H + \epsilon (\delta H + c\1_2) + \mc{O}(\epsilon^2) \,,\quad c\in\mb{R}\,,
\ed
induces the same linear perturbation to the state as in Lemma \ref{lem:Hpert}. We may then choose $c$ so that also the determinant of the perturbed Hamiltonian vanishes (again up to first order in $\epsilon$):
\bd
	\det(H + \epsilon (\delta H + c\1_2)) = 0 \,.
\ed
Using the identity $\det(X) = \frac{1}{2}(\tr(X)^2 - \tr(X^2))$ valid for 2-by-2 matrices, and the property $H^2 = \tr(H)H$, we may easily solve for $c$, and get
\bd
	c = \frac{\tr(H\delta H)}{\tr(H)} - \tr(\delta H) \,.
\ed

The remaining task is then to find an infinitesimal Lorentz transformation generated by some $Z\in \mb{M}_2(\mb{C})$, $\tr(Z)=0$, such that $e^{\epsilon Z} H e^{\epsilon Z^*} = H + \epsilon(\delta H + c\1_2) + \mc{O}(\epsilon^2)$, which yields the requirement
\bd
	ZH + HZ^* = \delta H + c\1_2 \,.
\ed
If we split $Z$ into its hermitian and anti-hermitian parts as $Z = Z_+ + Z_-$, where $Z_\pm^* = \pm Z_\pm$, we get
\bd
	[Z_+,H]_+ + [Z_-,H]_- = \delta H + c\1_2 \,,
\ed
where $[X,Y]_\pm = XY \pm YX$. The hermitian part of $Z$ generates a pure boost while the anti-hermitian part generates a spatial rotation. Since they commute at the linear level, we may solve for both separately. The general idea of the derivation is that we first use $Z_+$ to boost $H$ so that its trace agrees with the trace of $\delta H + c\1_2$, and we can then rotate to $\delta H + c\1_2$ using $Z_-$.

Let us first focus on $Z_+$. We start with the ansatz $Z_+ = r(H-\frac{1}{2}\tr(H)\1_2)$ for some $r\in\mb{R}$, which generates a boost to the direction of $H$. As said, we require
\bd
	\tr([Z_+,H]_+) = 2r\tr(H^2-\frac{1}{2}\tr(H)H) = \tr(\delta H + c\1_2) \,,
\ed
from which we can easily solve
\bd
	r = \frac{\tr(\delta H) + 2c}{\tr(H)^2} \,.
\ed

The next task is then to find $Z_-$ such that
\bea
	[Z_-,H]_- & = & \delta H + c\1_2 - [Z_+,H]_+ \nonumber\\
	& = & (\delta H - \frac{1}{2}\tr(\delta H) \1_2) + \lambda(H - \frac{1}{2}\tr(H)\1_2) \,, \nonumber
\eea
where the last line follows by substitution, and denoting
\bd
 \lambda := \frac{2\tr(H\delta H) - \tr(H)\tr(\delta H)}{\tr(H)^2} \,.
\ed
To that end, let us introduce another elementary lemma:
\begin{lemma}\label{lem:tripleprod}
Let $X_i\in\mb{M}_2(\mb{C})$, $i=1,2,3$. Then,
\bd
	\frac{1}{2}[[X_1,X_2],X_3] = \emph{\tr}(X_2^o X_3^o) X_1^o - \emph{\tr}(X_3^o X_1^o) X_2^o \,,
\ed
where $X_i^o := X_i - \frac{1}{2}\emph{\tr}(X_i)\emph{\1}_2$ is the trace-free part of $X_i$.
\end{lemma}
\begin{proof}
Denote $X_i^o = \vec{x}_i \cdot \vec{\sigma}$, where $\vec{x}_i \in \mb{R}^3$, $\vec{\sigma} = (\sigma_1, \sigma_2, \sigma_3)$, and $\sigma_k$ are the Pauli matrices. Then, the statement follows directly from the identities $[X_i,X_j] = 2i(\vec{x}_i \wedge \vec{x}_j) \cdot \vec{\sigma}$, $\tr(X_i^o X_j^o) = 2(\vec{x}_i \cdot \vec{x}_j)$, and the well-known triple product identity for vectors in $\mb{R}^3$
\bd
	(\vec{x}_1 \wedge \vec{x}_2) \wedge \vec{x}_3 = (\vec{x}_1 \cdot \vec{x}_3) \vec{x}_2 - (\vec{x}_2 \cdot \vec{x}_3) \vec{x}_1 \,. \qedhere
\ed
\end{proof}
Using Lemma \ref{lem:tripleprod} it is straightforward to verify that choosing $Z_- = \frac{1}{\tr(H)^2}[\delta H, H]$ yields exactly the desired result. Thus, we find for $Z$ the expression
\bea
	Z & = & Z_+ + Z_- = r(H-\frac{1}{2}\tr(H)\1_2) + \frac{1}{\tr(H)^2}[\delta H, H] \nonumber\\
	& = & \frac{1}{\tr(H)^2}\left(\tau(H - \frac{1}{2}\tr(H)\1_2) + [\delta H, H]\right) \,,
\eea
where $\tau := (2\tr(H\delta H) - \tr(H)\tr(\delta H))/\tr(H)$, as stated in the main text.

%\bibliography{LCQIbib}

\begin{thebibliography}{19}%
\makeatletter
\providecommand \@ifxundefined [1]{%
 \@ifx{#1\undefined}
}%
\providecommand \@ifnum [1]{%
 \ifnum #1\expandafter \@firstoftwo
 \else \expandafter \@secondoftwo
 \fi
}%
\providecommand \@ifx [1]{%
 \ifx #1\expandafter \@firstoftwo
 \else \expandafter \@secondoftwo
 \fi
}%
\providecommand \natexlab [1]{#1}%
\providecommand \enquote  [1]{``#1''}%
\providecommand \bibnamefont  [1]{#1}%
\providecommand \bibfnamefont [1]{#1}%
\providecommand \citenamefont [1]{#1}%
\providecommand \href@noop [0]{\@secondoftwo}%
\providecommand \href [0]{\begingroup \@sanitize@url \@href}%
\providecommand \@href[1]{\@@startlink{#1}\@@href}%
\providecommand \@@href[1]{\endgroup#1\@@endlink}%
\providecommand \@sanitize@url [0]{\catcode `\\12\catcode `\$12\catcode
  `\&12\catcode `\#12\catcode `\^12\catcode `\_12\catcode `\%12\relax}%
\providecommand \@@startlink[1]{}%
\providecommand \@@endlink[0]{}%
\providecommand \url  [0]{\begingroup\@sanitize@url \@url }%
\providecommand \@url [1]{\endgroup\@href {#1}{\urlprefix }}%
\providecommand \urlprefix  [0]{URL }%
\providecommand \Eprint [0]{\href }%
\providecommand \doibase [0]{http://dx.doi.org/}%
\providecommand \selectlanguage [0]{\@gobble}%
\providecommand \bibinfo  [0]{\@secondoftwo}%
\providecommand \bibfield  [0]{\@secondoftwo}%
\providecommand \translation [1]{[#1]}%
\providecommand \BibitemOpen [0]{}%
\providecommand \bibitemStop [0]{}%
\providecommand \bibitemNoStop [0]{.\EOS\space}%
\providecommand \EOS [0]{\spacefactor3000\relax}%
\providecommand \BibitemShut  [1]{\csname bibitem#1\endcsname}%
\let\auto@bib@innerbib\@empty
%</preamble>
\bibitem [{\citenamefont {Wald}(1984)}]{Wald84}%
  \BibitemOpen
  \bibfield  {author} {\bibinfo {author} {\bibfnamefont {R.}~\bibnamefont
  {Wald}},\ }\href@noop {} {\emph {\bibinfo {title} {General Relativity}}}\
  (\bibinfo  {publisher} {University of Chicago Press},\ \bibinfo {address}
  {Chicago},\ \bibinfo {year} {1984})\BibitemShut {NoStop}%
\bibitem [{\citenamefont {Haag}(1996)}]{Haag96}%
  \BibitemOpen
  \bibfield  {author} {\bibinfo {author} {\bibfnamefont {R.}~\bibnamefont
  {Haag}},\ }\href@noop {} {\emph {\bibinfo {title} {Local Quantum Physics:
  Fields, Particles, Algebras}}}\ (\bibinfo  {publisher} {Springer},\ \bibinfo
  {address} {Berlin Heidelberg},\ \bibinfo {year} {1996})\BibitemShut {NoStop}%
\bibitem [{\citenamefont {Wigner}(1939)}]{Wigner39}%
  \BibitemOpen
  \bibfield  {author} {\bibinfo {author} {\bibfnamefont {E.}~\bibnamefont
  {Wigner}},\ }\href {\doibase 10.2307/1968551} {\bibfield  {journal} {\bibinfo
   {journal} {Ann. Math.}\ }\textbf {\bibinfo {volume} {40}},\ \bibinfo {pages}
  {149} (\bibinfo {year} {1939})}\BibitemShut {NoStop}%
\bibitem [{\citenamefont {Bekenstein}(1973)}]{Bekenstein73}%
  \BibitemOpen
  \bibfield  {author} {\bibinfo {author} {\bibfnamefont {J.~D.}\ \bibnamefont
  {Bekenstein}},\ }\href {\doibase 10.1103/PhysRevD.7.2333} {\bibfield
  {journal} {\bibinfo  {journal} {Phys. Rev.}\ }\textbf {\bibinfo {volume}
  {D7}},\ \bibinfo {pages} {2333} (\bibinfo {year} {1973})}\BibitemShut
  {NoStop}%
%%CITATION = PHRVA,D7,2333;%%
\bibitem [{\citenamefont {Carlip}(2014)}]{Carlip14}%
  \BibitemOpen
  \bibfield  {author} {\bibinfo {author} {\bibfnamefont {S.}~\bibnamefont
  {Carlip}},\ }\href {\doibase 10.1142/S0218271814300237} {\bibfield  {journal}
  {\bibinfo  {journal} {Int. J. Mod. Phys.}\ }\textbf {\bibinfo {volume}
  {D23}},\ \bibinfo {pages} {1430023} (\bibinfo {year} {2014})},\ \Eprint
  {http://arxiv.org/abs/1410.1486} {arXiv:1410.1486 [gr-qc]} \BibitemShut
  {NoStop}%
%%CITATION = ARXIV:1410.1486;%%
\bibitem [{\citenamefont {Brunetti}\ \emph {et~al.}(2003)\citenamefont
  {Brunetti}, \citenamefont {Fredenhagen},\ and\ \citenamefont
  {Verch}}]{BrunettiFredenhagenVerch03}%
  \BibitemOpen
  \bibfield  {author} {\bibinfo {author} {\bibfnamefont {R.}~\bibnamefont
  {Brunetti}}, \bibinfo {author} {\bibfnamefont {K.}~\bibnamefont
  {Fredenhagen}}, \ and\ \bibinfo {author} {\bibfnamefont {R.}~\bibnamefont
  {Verch}},\ }\href {\doibase 10.1007/s00220-003-0815-7} {\bibfield  {journal}
  {\bibinfo  {journal} {Commun. Math. Phys.}\ }\textbf {\bibinfo {volume}
  {237}},\ \bibinfo {pages} {31} (\bibinfo {year} {2003})},\ \Eprint
  {http://arxiv.org/abs/math-ph/0112041} {arXiv:math-ph/0112041} \BibitemShut
  {NoStop}%
\bibitem [{\citenamefont {Buchholz}\ and\ \citenamefont
  {Summers}(2005)}]{Buchholz04b}%
  \BibitemOpen
  \bibfield  {author} {\bibinfo {author} {\bibfnamefont {D.}~\bibnamefont
  {Buchholz}}\ and\ \bibinfo {author} {\bibfnamefont {S.~J.}\ \bibnamefont
  {Summers}},\ }\href {\doibase 10.1016/j.physleta.2005.01.055} {\bibfield
  {journal} {\bibinfo  {journal} {Phys. Lett.}\ }\textbf {\bibinfo {volume}
  {A337}},\ \bibinfo {pages} {17} (\bibinfo {year} {2005})},\ \Eprint
  {http://arxiv.org/abs/quant-ph/0403149} {arXiv:quant-ph/0403149} \BibitemShut
  {NoStop}%
%%CITATION = QUANT-PH/0403149;%%
\bibitem [{\citenamefont {Buchholz}\ \emph {et~al.}(1987)\citenamefont
  {Buchholz}, \citenamefont {D'Antoni},\ and\ \citenamefont
  {Fredenhagen}}]{Buchholz87}%
  \BibitemOpen
  \bibfield  {author} {\bibinfo {author} {\bibfnamefont {D.}~\bibnamefont
  {Buchholz}}, \bibinfo {author} {\bibfnamefont {C.}~\bibnamefont {D'Antoni}},
  \ and\ \bibinfo {author} {\bibfnamefont {K.}~\bibnamefont {Fredenhagen}},\
  }\href {\doibase 10.1007/BF01239019} {\bibfield  {journal} {\bibinfo
  {journal} {Commun. Math. Phys.}\ }\textbf {\bibinfo {volume} {111}},\
  \bibinfo {pages} {123} (\bibinfo {year} {1987})}\BibitemShut {NoStop}%
\bibitem [{\citenamefont {Pye}\ \emph {et~al.}(2015)\citenamefont {Pye},
  \citenamefont {Donnelly},\ and\ \citenamefont {Kempf}}]{Pye15}%
  \BibitemOpen
  \bibfield  {author} {\bibinfo {author} {\bibfnamefont {J.}~\bibnamefont
  {Pye}}, \bibinfo {author} {\bibfnamefont {W.}~\bibnamefont {Donnelly}}, \
  and\ \bibinfo {author} {\bibfnamefont {A.}~\bibnamefont {Kempf}},\ }\href
  {\doibase 10.1103/PhysRevD.92.105022} {\bibfield  {journal} {\bibinfo
  {journal} {Phys. Rev.}\ }\textbf {\bibinfo {volume} {D92}},\ \bibinfo {pages}
  {105022} (\bibinfo {year} {2015})},\ \Eprint
  {http://arxiv.org/abs/1508.05953} {arXiv:1508.05953 [quant-ph]} \BibitemShut
  {NoStop}%
%%CITATION = ARXIV:1508.05953;%%
\bibitem [{\citenamefont {Bousso}(1999)}]{Bousso99}%
  \BibitemOpen
  \bibfield  {author} {\bibinfo {author} {\bibfnamefont {R.}~\bibnamefont
  {Bousso}},\ }\href {\doibase 10.1088/1126-6708/1999/07/004} {\bibfield
  {journal} {\bibinfo  {journal} {JHEP}\ }\textbf {\bibinfo {volume} {07}},\
  \bibinfo {pages} {004} (\bibinfo {year} {1999})},\ \Eprint
  {http://arxiv.org/abs/hep-th/9905177} {arXiv:hep-th/9905177} \BibitemShut
  {NoStop}%
%%CITATION = HEP-TH/9905177;%%
\bibitem [{\citenamefont {Srednicki}(1993)}]{Srednicki93}%
  \BibitemOpen
  \bibfield  {author} {\bibinfo {author} {\bibfnamefont {M.}~\bibnamefont
  {Srednicki}},\ }\href {\doibase 10.1103/PhysRevLett.71.666} {\bibfield
  {journal} {\bibinfo  {journal} {Phys. Rev. Lett.}\ }\textbf {\bibinfo
  {volume} {71}},\ \bibinfo {pages} {666} (\bibinfo {year} {1993})},\ \Eprint
  {http://arxiv.org/abs/hep-th/9303048} {arXiv:hep-th/9303048} \BibitemShut
  {NoStop}%
%%CITATION = HEP-TH/9303048;%%
\bibitem [{\citenamefont {Casini}\ \emph {et~al.}(2015)\citenamefont {Casini},
  \citenamefont {Mazzitelli},\ and\ \citenamefont {Test\'e}}]{Casini14}%
  \BibitemOpen
  \bibfield  {author} {\bibinfo {author} {\bibfnamefont {H.}~\bibnamefont
  {Casini}}, \bibinfo {author} {\bibfnamefont {F.~D.}\ \bibnamefont
  {Mazzitelli}}, \ and\ \bibinfo {author} {\bibfnamefont {E.}~\bibnamefont
  {Test\'e}},\ }\href {\doibase 10.1103/PhysRevD.91.104035} {\bibfield
  {journal} {\bibinfo  {journal} {Phys. Rev.}\ }\textbf {\bibinfo {volume}
  {D91}},\ \bibinfo {pages} {104035} (\bibinfo {year} {2015})},\ \Eprint
  {http://arxiv.org/abs/1412.6522} {arXiv:1412.6522 [hep-th]} \BibitemShut
  {NoStop}%
%%CITATION = ARXIV:1412.6522;%%
\bibitem [{\citenamefont {Ojima}(1986)}]{Ojima86}%
  \BibitemOpen
  \bibfield  {author} {\bibinfo {author} {\bibfnamefont {I.}~\bibnamefont
  {Ojima}},\ }\href {\doibase 10.1007/BF00417467} {\bibfield  {journal}
  {\bibinfo  {journal} {Lett. Math. Phys.}\ }\textbf {\bibinfo {volume} {11}},\
  \bibinfo {pages} {73} (\bibinfo {year} {1986})}\BibitemShut {NoStop}%
\bibitem [{\citenamefont {G{\"o}rnitz}\ \emph {et~al.}(1992)\citenamefont
  {G{\"o}rnitz}, \citenamefont {Graudenz},\ and\ \citenamefont
  {v.~Weizs{\"a}cker}}]{Goernitz92}%
  \BibitemOpen
  \bibfield  {author} {\bibinfo {author} {\bibfnamefont {T.}~\bibnamefont
  {G{\"o}rnitz}}, \bibinfo {author} {\bibfnamefont {D.}~\bibnamefont
  {Graudenz}}, \ and\ \bibinfo {author} {\bibfnamefont {C.~F.}\ \bibnamefont
  {v.~Weizs{\"a}cker}},\ }\href {\doibase 10.1007/BF00671965} {\bibfield
  {journal} {\bibinfo  {journal} {Int. J. Theor. Phys.}\ }\textbf {\bibinfo
  {volume} {31}},\ \bibinfo {pages} {1929} (\bibinfo {year}
  {1992})}\BibitemShut {NoStop}%
\bibitem [{\citenamefont {H\"ohn}\ and\ \citenamefont
  {M\"uller}(2016)}]{Hoehn14}%
  \BibitemOpen
  \bibfield  {author} {\bibinfo {author} {\bibfnamefont {P.~A.}\ \bibnamefont
  {H\"ohn}}\ and\ \bibinfo {author} {\bibfnamefont {M.~P.}\ \bibnamefont
  {M\"uller}},\ }\href {\doibase 10.1088/1367-2630/18/6/063026} {\bibfield
  {journal} {\bibinfo  {journal} {New J. Phys.}\ }\textbf {\bibinfo {volume}
  {18}},\ \bibinfo {pages} {063026} (\bibinfo {year} {2016})},\ \Eprint
  {http://arxiv.org/abs/1412.8462} {arXiv:1412.8462 [quant-ph]} \BibitemShut
  {NoStop}%
%%CITATION = ARXIV:1412.8462;%%
\bibitem [{\citenamefont {Raasakka}(2017)}]{Raasakka16}%
  \BibitemOpen
  \bibfield  {author} {\bibinfo {author} {\bibfnamefont {M.}~\bibnamefont
  {Raasakka}},\ }\href {\doibase 10.3842/SIGMA.2017.006} {\bibfield  {journal}
  {\bibinfo  {journal} {SIGMA}\ }\textbf {\bibinfo {volume} {13}},\ \bibinfo
  {pages} {006} (\bibinfo {year} {2017})},\ \Eprint
  {http://arxiv.org/abs/1605.03942} {arXiv:1605.03942 [gr-qc]} \BibitemShut
  {NoStop}%
%%CITATION = ARXIV:1605.03942;%%
\bibitem [{\citenamefont {Jacobson}(1995)}]{Jacobson95}%
  \BibitemOpen
  \bibfield  {author} {\bibinfo {author} {\bibfnamefont {T.}~\bibnamefont
  {Jacobson}},\ }\href {\doibase 10.1103/PhysRevLett.75.1260} {\bibfield
  {journal} {\bibinfo  {journal} {Phys. Rev. Lett.}\ }\textbf {\bibinfo
  {volume} {75}},\ \bibinfo {pages} {1260} (\bibinfo {year} {1995})},\ \Eprint
  {http://arxiv.org/abs/gr-qc/9504004} {arXiv:gr-qc/9504004} \BibitemShut
  {NoStop}%
%%CITATION = GR-QC/9504004;%%
\bibitem [{\citenamefont {Jacobson}(2016)}]{Jacobson15}%
  \BibitemOpen
  \bibfield  {author} {\bibinfo {author} {\bibfnamefont {T.}~\bibnamefont
  {Jacobson}},\ }\href {\doibase 10.1103/PhysRevLett.116.201101} {\bibfield
  {journal} {\bibinfo  {journal} {Phys. Rev. Lett.}\ }\textbf {\bibinfo
  {volume} {116}},\ \bibinfo {pages} {201101} (\bibinfo {year} {2016})},\
  \Eprint {http://arxiv.org/abs/1505.04753} {arXiv:1505.04753 [gr-qc]}
  \BibitemShut {NoStop}%
%%CITATION = ARXIV:1505.04753;%%
\bibitem [{\citenamefont {Hall}(2003)}]{Hall03}%
  \BibitemOpen
  \bibfield  {author} {\bibinfo {author} {\bibfnamefont {B.~C.}\ \bibnamefont
  {Hall}},\ }\href@noop {} {\emph {\bibinfo {title} {Lie Groups, Lie Algebras,
  and Representations: An Elementary Introduction}}}\ (\bibinfo  {publisher}
  {Springer},\ \bibinfo {address} {New York},\ \bibinfo {year}
  {2003})\BibitemShut {NoStop}%
\end{thebibliography}


\end{document}